\documentclass[anonymous]{article}
\usepackage{ijcai26}

\usepackage{times}
\usepackage{soul}
\usepackage{url}
\usepackage[hidelinks]{hyperref}
\usepackage[utf8]{inputenc}
\usepackage[small]{caption}
\usepackage{graphicx}
\usepackage{amsmath}
\usepackage{amsthm}
\usepackage{booktabs}
\usepackage{algorithm}
\usepackage{algorithmic}
\usepackage[switch]{lineno}

\usepackage{balance} 

\usepackage{amsfonts}
\usepackage{comment}
\usepackage{thmtools}
\usepackage{mathtools}
\usepackage{tikz}
\usepackage{array}
\usepackage{caption}
\usepackage
[
    subrefformat = simple, 
    labelformat = simple,  
]
{subcaption}         
\usepackage{wrapfig} 
\usepackage{cleveref}

\newtheorem{theorem}{Theorem}

\newtheorem{lemma}[theorem]{Lemma}
\newtheorem{corollary}[theorem]{Corollary}
\newtheorem{proposition}[theorem]{Proposition}
\newtheorem{definition}[theorem]{Definition}

\title{Temporal Network Creation Games: The Impact of Flexible Labels}

\author{
Hans Gawendowicz$^1$
\and
Nicolas Klodt$^1$\and
Aleksandrs Morgensterns$^1$\And
George Skretas$^1$\\
\affiliations
$^1$Hasso Plattner Institute, University of Potsdam, Germany\\
\emails
\{hans.gawendowicz, nicolas.klodt\}@hpi.de,
aleksandrs.morgensterns@gmail.com,
georgios.skretas@hpi.de
}

\begin{document}

\maketitle

\begin{abstract}


A crucial aspect of research is understanding how real-world networks, such as transportation and information networks, are formed. 
A prominent model for such networks was introduced by \cite{fabrikant_network_2003} and extended by \cite{bilo_temporal_2023}, incorporating temporal graphs to better represent real-world networks. In this model, there is a given host graph with $n$ agents (represented by nodes) and time labels on the edges. Each agent can establish connections by purchasing edges. This makes the edges present at the time steps given by the time labels of the host graph. The goal of each agent is to reach as many other agents as possible while minimizing the number of edges bought. 
However, this model makes the simplifying assumption that each edge comes with predetermined time steps. We address this deficiency by extending the model of Bilo et al. \cite{bilo_temporal_2023} to allow agents to purchase edges and to decide when they appear. To capture a variety of real-world applications, we study two reachability models and several cost functions based on the label an agent assigns to an edge. For these settings, we provide proofs of existence of Nash equilibria, as well as lower and upper bounds on the Price of Anarchy and Price of Stability.
\end{abstract}

\section{Introduction}

In this paper, we aim to investigate the formation of logistics/transportation networks and information networks. In particular, we are interested in the subset of these networks formed by the actions of independent agents. An example of such a logistics network is the supply chain of companies involved in the production of goods. Every company has a factory and is interested in transporting goods from it to other locations. Their primary interest in these connections is to ensure that goods from their factory can reach all other locations in the network while minimizing transportation costs. For an example of information networks, consider the communication network in a company. The company's personnel must conduct meetings so that the necessary information reaches the relevant people. However, working time is important, and they want to minimize the time spent in meetings.

Several models have been considered for studying these types of networks \cite{Anshelevich_2004,Bala_2000,Jackson_1996}. Their main key ingredient is that they assume that each node of the network is a separate entity, a selfish and non-cooperative agent. One of the most studied models is the \emph{Network Creation Game} (NCG)~\cite{fabrikant_network_2003}. In this model, selfish agents act as nodes in a network that can form costly connections to others to gain a central position within the emerging network. The goal of the agents is to maximize the number of other nodes they can reach while minimizing the cost they incur. Early studies in this model assume that a network's connections are static, i.e., always available, which is unrealistic given the examples we have considered so far. Recently, a generalization of this model, called \emph{Temporal Network Creation Games}, has been introduced \cite{bilo_temporal_2023}, where it is assumed that the agents can buy connections that are available only in a particular moment in time, i.e., a plane flies from Frankfurt to London with departure time at 3 pm. 

In this line of work \cite{bilo_temporal_2023,bilo_temporal_2024} , the authors have studied multiple variants of this model. The key differences between these variants refer to whether (i) the agents can buy only incident connections or any connection in the network, (ii) there are multiple connections that can be bought between two nodes of the network, and (iii) the agents want to reach every other node in the graph or just a subset of them. However, for every one of these variants, the model made the simplifying assumption that every node must buy connections at predetermined times, and they are not allowed to pick them. This is unrealistic for the same scenarios we want to study, as a logistics company would be allowed to schedule its vehicles to run at whatever time suits it best. Our goal in this paper is to generalize the model introduced by Bilo et. al \cite{bilo_temporal_2023}, to allow agents to buy a connection and decide at which point in time it is available. After having done that, we will study the behavior of this model.

\subsection{Our Contribution}
We introduce a new model for Temporal Network Creation Games, where the network is modeled as an empty graph, with each vertex representing a selfish agent. Every agent can buy an edge between themselves and another agent and decide at which point in time the edge is available by assigning a \emph{label} to it. The goal of each agent is to reach a subset of the graph's vertices while minimizing the cost of the edges they buy. Since we aim to model both transportation/logistics and information networks, we must consider how reachability operates in each. In both cases, a vertex $u$ can reach a vertex $v$ if there exists a path of consecutive edges that have increasing labels. Since the transportation of goods is not instantaneous, it is more appropriate to assume that labels must be strictly increasing, whereas in information networks, non-strictly increasing labels can suffice. 

We also have to consider how the cost of an edge relates to the label assigned to it. For example, the cost of transporting goods may not depend on the time, or in a company, late meetings might be discouraged. To account for this, we consider different cost functions. In particular, we consider four different label cost functions: (i) a label cost function where every label costs the same, (ii) two related label cost functions where earlier/later labels cost more, (iii) a label cost function where every label costs the same but no two adjacent edges can have the same label and (iv) a label cost function where every label costs the same but there is a lower bound on the possible label values that can be assigned.

For most combinations of reachability model and label cost function, we provide proofs for the existence of Nash equilibria and bounds on the Price of Anarchy (PoA) and Price of Stability (PoS). The bounds and existence results are summarized in Table \ref{tab:placeholder}. The paper is organized as follows. In Section \ref{preliminaries}, we define the model and different variants we are exploring. We then dedicate one section to each label cost function and, within it, analyze both reachability models. In Section \ref{chap:flexiblelabel}, we analyze the Uniform Label cost function, where the cost of all labels is the same. In Section 
\ref{sec:monotone}, we analyze the Monotone Label cost function, where the cost depends on the particular value that a label has. In Section \ref{sec:proper}, we analyze the Proper Label cost function, where the cost of all labels is the same, but no two adjacent edges can have the same label value. In Section \ref{sec:monotone}, we analyze the Arbitrary Low Label function, where the cost of all labels is the same, but there is a lower bound on the value of a label. 

\begin{table*}[]
    \centering
\begin{tabular}{||l|c c c|c c c||} 
 \hline
 & \multicolumn{3}{c|}{Non-strict Paths} & \multicolumn{3}{c||}{Strict Paths} \\ 
 \cline{2-7}
 & PoS & PoA & Existence & PoS & PoA & Existence \\
 \hline\hline
 Any Label 
 & $1$ [\ref{prop:posuniform}] & $2 - \mathcal{O}\left(\frac{1}{\sqrt{n}}\right)$ [\ref{thm:unilabel2}] & Yes [\ref{prop:posuniform}]
 & $1$ [\ref{thm:unilabpos}] & $\Theta(n)$ [\ref{thm:anystrictPoA}] & Yes [\ref{thm:anystrictPoA}] \\ 
 \hline
 Label Cost $f_\mathbf{s}^\downarrow$ 
 & $1$ [\ref{thm:poaposfdown}] & $1$ [\ref{thm:poaposfdown}] & Yes [\ref{thm:poaposfdown}]
 & $\Theta(n)$ [\ref{strictdownpos}] & $\Theta(n)$ [\ref{strictmonotonepoa}] & Yes [\ref{strictmonotonepoa}] \\ 
 \hline
 Label Cost $f_\mathbf{s}^\uparrow$ 
 & $1$ [\ref{nonstrictposup}] & $[2 - \mathcal{O}\left(\frac{1}{\sqrt{n}}\right), 3]$ [\ref{monotonepoaup}] & Yes [\ref{nonstrictposup}]
 & ? &  $\Theta(n)$ [\ref{strictmonotonepoa}] & Yes [\ref{strictmonotonepoa}] \\ 
 \hline
 Proper Labels 
 & $1$ [\ref{propernonstrict}] & $\Omega(\log n)$ [\ref{propernonstrict}] & Yes [\ref{propernonstrict}]
 & $1$ [\ref{properpos}] & $\Omega(\log n)$ [\ref{properpoa}] & Yes [\ref{properpos}] \\ 
 \hline
 Arbitrary Low Labels 
 & $1$ [\ref{nonstrictarbitrary}] & $2 - \mathcal{O}\left(\frac{1}{\sqrt{n}}\right)$ [\ref{nonstrictarbitrary}] & Yes [\ref{nonstrictarbitrary}]
 & $1^*$ [\ref{arbitrarypos}] & $1 + \mathcal{O}\left(\frac{1}{n}\right)^*$ [\ref{arbitrarypoa}] & ? \\ 
 \hline
\end{tabular}

    \caption{PoS and PoA bounds and existence of NE. *For strict path arbitrarily low labels the results apply only for $n \leq 6$.}
    \label{tab:placeholder}
\end{table*}

\subsection{Related Work}
This paper is part of a recent line of work combining two formerly separate areas of research: Network Creation Games and Temporal Networks. 

Network Creation Games were pioneered in \cite{fabrikant_network_2003} on static graphs. There, agents act as nodes in a network that can buy costly edges to other agents to minimize the distance to all other agents while keeping the cost of buying edges as low as possible. This sparked a plethora of studies on network creation games, including improving bounds on the PoA \cite{poa_mostly_constant_2013}, pursuing the famous tree conjecture \cite{bilo_tree_conjecture_2020,dippel_2022}, and analyzing different degrees of cooperation \cite{demaine_price_2009,friedrich_2023}, but also introducing various variants to make the initially very simplistic model more realistic. This includes introducing a host graph restricting the edges that can be bought \cite{demaine_price_2009} or even placing the agents onto an underlying geometric graph with distance metrics as edge weights, as in \cite{abam_geometric_2019} and \cite{bilo_geometric_2020}. While this explosion of research efforts opened many new directions, dynamic graphs were only recently introduced into the field by \cite{bilo_temporal_2023}, formulating the original temporal network creation game model. They proved the existence of Nash equilibria for lifetime $t \leq 2$ and showed upper bounds on the PoA for Nash and Greedy equilibria. \cite{bilo_temporal_2024} later formalized the terminal model, showed the existence of equilibria for number of terminals $|T| \leq 2$, and analyzed non-local edge buying.

Since the size of the worst Nash equilibria in this temporal network creation model is bounded upwards by the largest minimal temporal spanner, equilibria in temporal network creation games are closely related to minimal spanners in temporal graphs. \cite{kempe_connectivity_2002} introduced temporal graphs and constructed a minimal spanner with $\Theta(n\log n)$ edges. \cite{axiotis_size_2016} improved the upper bound to $\Theta(n^2)$. \cite{casteigts_temporal_2021} showed that any complete and connected simple (only one label per edge) temporal graph admits a temporal spanner with at most $\mathcal{O}(n\log n)$ edges. \cite{angrick_how_2024} recently advanced the result by showing that edge-pivotable graphs admit linear-size spanners. Whether all complete connected simple temporal graphs have linear-sized spanners remains an open question. For temporal graphs with multiple labels per edge \cite{christiann_inefficiently_2023} showed that there are connected graphs where all labels are needed for connectivity. Algorithmically, \cite{axiotis_size_2016} showed that selecting a smallest subset of given edges in a temporal graph to make it connected is NP-hard, which is closely linked to deciding optimal strategy in temporal network creation games. Under certain conditions, they provided polynomial-time algorithms.

Other types of temporal graphs were also studied, such as random temporal graphs in \cite{casteigts_sharp_2023}. As most versions of problems on static graphs focus on the distance between nodes, the study of low-stretch temporal graphs as in \cite{bilo_sparse_2022} can be of interest. Results can be used to extend the model by including distance in the cost function. Different types of temporal reachability were also studied in \cite{bilo_temporal_2023}, as well as reachability graphs in \cite{whitbeck_temporal_2012}.


Our research distinguishes itself from current results by being the first to examine other goals (cost functions) for the vertices, and by separating the edge labeling from the host graph by making the choice of the time label part of the agents' strategic decisions.

\section{Preliminaries}
\label{preliminaries}
We start by defining key terms in temporal graph and game theory that are necessary for our results.
\subsection{Temporal Graphs}
    A \emph{temporal graph} $G = (V_G, E_G, \lambda_G)$ is an undirected labeled graph with the labeling function $\lambda_G\colon E_G \mapsto \mathbb{Z}_{\leq t}$. We call $t$ the \emph{lifetime} of $G$. Further, a labeling is called proper if no 2 adjacent edges have the same label. We will sometimes refer to an edge with label $k$ as $k$-edge for brevity.
    
    A \emph{temporal path} in $G$ is defined as a sequence of edges $(e_1, \dots, e_i) \in E_G$ such that $\forall 1\leq k\leq i-1\colon \lambda_G(e_k) \leq \lambda_G(e_{k+1})$. A \emph{strict temporal path}, sometimes referred to as a strict journey, is defined as a temporal path where $\forall 1\leq k\leq i-1: \lambda_G(e_k) < \lambda_G(e_{k+1})$. For any vertex $v\in V_G$ we define $R_G(v)$ as the set of vertices which can be reached via a temporal path from $v$. Unless obvious, we denote $R_G(v)^\leq$ as non-strict reachability and $R_G(v)^<$ as strict reachability. For brevity, we use $R_G(v)$ after establishing which reachability we are using for each model. If $\forall v\in V_G\colon R_G(v) = V_G$, we say $G$ is \emph{temporally connected}. A temporally connected subgraph $G' = (V_{G'}, E_{G'}, \lambda_G|_{E_{G'}})$ with $V_{G'} = V_G$ is called a \emph{temporal spanner} (of $G$). Furthermore, for $t \in \mathbb{N}$, let $R_{G, t}(v)$ denote  the set of vertices $v$ can reach using temporal paths $(e_1, \dots, e_i)$ with $\lambda_G(e_1) \geq t$.

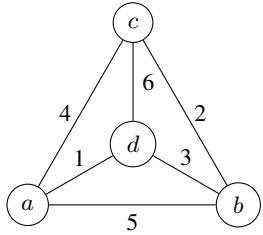
\begin{figure}
    \centering
        \begin{tikzpicture}[font=\small, scale=0.7]
            \node[circle, draw] (A) at (0,0) {$a$};
            \node[circle, draw] (B) at (4,0) {$b$};
            \node[circle, draw] (C) at (2,3.464) {$c$};
            \node[circle, draw] (E) at (2,1.155) {$d$};

            \draw (A) -- node[below] {5} (B);
            \draw (B) -- node[right] {2} (C);
            \draw (C) -- node[left] {4} (A);
            \draw (E) -- node[above] {1} (A);
            \draw (E) -- node[above] {3} (B);
            \draw (E) -- node[right] {6} (C);
        \end{tikzpicture}
    \caption{Temporal graph with lifetime 6. $(\{a, d\}, \{d, b\})$ is a temporal path, while $(\{a,c\}, \{c, b\})$ is not.}
    \label{tempgraphexamp}
\end{figure}
An example for a temporal graph can be seen in \Cref{tempgraphexamp}. We now define structures that will be useful in later proofs. 
\begin{definition}[$k$-label tree]
We call a temporal graph $G = (V_G, E_G, \lambda_G)$ a \emph{$k$-label tree} if $(V_G, E_G)$ is a tree and $\forall e \in E_G: \lambda_G(e) = k$.
\end{definition}
\begin{definition}[Reachability Tree]
    Let $G = (V_G, E_G, \lambda_G)$ be a temporal graph and let $v \in V_G$. A \emph{reachability tree} $\mathcal{T}_G(v)$ rooted at $v$ is a temporal subgraph of $G$ such that $\mathcal{T}_G(v)$ is a tree and $R_{\mathcal{T}_G(v)}(v) = R_G(v)$.
\end{definition}
\subsection{Our Model}
Next, we introduce our game-theoretic model. An instance of the game consists of a set of agents $V$. Each agent $v\in V$ chooses a strategy $S_v\subseteq V\times \mathbb{Z}$ which encodes to which other agents and at which points in time $v$ wants to form a connection to.
The strategy profile $\mathbf{s}=\bigcup_{v\in V} \{(v,S_v)\}$ is the combination of the strategies of all agents together. By playing their strategies, the agents form a temporal graph

\begin{equation*}
    G(\mathbf{s}) = \left( V, \bigcup_{u\in V}\bigcup_{(v, l) \in S_u}\{u, v\}, \lambda_{G(\mathbf{s})} \right)\end{equation*}
with
\begin{equation*}
    \lambda_{G(\mathbf{s})}(\{u, v\}) = \min\{l \in \mathbb{N} | (v,l) \in S_u \lor (u, l) \in S_v\}
\end{equation*}
In short, we write $G(\mathbf{s}) = (V_\mathbf{s}, E_\mathbf{s}, \lambda_\mathbf{s})$. We take the minimal label bought between two vertices as both can only profit from reaching the other at an earlier time. This also disincentivizes the agents from buying multiple labels, meaning we can analyze only graphs with a single label per edge.\\
We now introduce a cost function for a given strategy profile and a given agent. Let $f_\mathbf{s}: E_{G(\mathbf{s})} \mapsto \left[0, \frac{1}{|V|}\right)$ be a label cost function that assigns cost to the label of a bought edge. Let also $P_\mathbf{s}(v)$ be some penalty function that punishes behaviors between agents that we want to discourage
Finally, $R_{G(\mathbf{s})}(v)$ is the reachability function of $v$ that can be either strict or non-strict. We calculate the cost for each agent using the number of edges bought, additional costs for labels chosen, penalties from a penalty function, and not being able to reach other agents. We define a cost function as
\begin{equation*}
    c(u, \mathbf{s}) = |S_u| + \smashoperator{\sum_{(v, l) \in S_u}}f_\mathbf{s}(\{u, v\}) + P_\mathbf{s}(u) + K\cdot (|V_{G(\mathbf{s})} \setminus R_{G(\mathbf{s})}(u)|)
\end{equation*}
with $K\geq 2$ as a sufficiently large constant. Note that the total cost an agent can save by minimizing the label cost is less than they can save by buying one less edge, meaning they will prioritize ensuring reachability and penalties, then number of edges and then label cost. We refer to a strategy profile $\mathbf{s}$ as a Nash Equilibrium (NE) if there is no agent $v$ such that there is a different strategy $S_v'$ for that agent with $c(v, \{\mathbf{s}\setminus (v,S_v)\} \cup (v, S_v')) < c(v, \mathbf{s})$. Further, we refer to $C(\mathbf{s}) = \sum_{v\in V_{G(\mathbf{s})}}c (u, \mathbf{s})$ as the social cost of $\mathbf{s}$ and the social optimum OPT as a strategy profile $\mathbf{s}$ which minimizes the social cost. We can now define the two metrics with which we can measure the quality of equilibria.

We define the Price of Anarchy (PoA) as the ratio of the social cost of the NE with the highest social cost and the social cost of the social optimum. Similarly, let the Price of Stability (PoS) be the ratio of the social cost of the NE with the lowest social cost and the social optimum. More formally, let $\text{NE}$ be the set of Nash Equilibria, then
\begin{align*}
    \text{PoA} &= \frac{\max_{\mathbf{s} \in \text{NE}} C(\mathbf{s})}{C(\text{OPT})}\quad\text{ and }\quad
    \text{PoS} &= \frac{\min_{\mathbf{s} \in \text{NE}} C(\mathbf{s})}{C(\text{OPT})}.
\end{align*}
\subsection{Variants}
We introduce the specific scenarios we will analyze. For the reachability function we always consider both strict and non-strict reachability. For the label cost function we define the following two.
We define
\begin{equation*}
    f_\mathbf{s}^\uparrow(e) = \frac{|\{e'\in E_{G(\mathbf{s})}, \lambda_{G(\mathbf{s})}(e') > \lambda_{G(\mathbf{s})}(e)\}|}{|E_{G(\mathbf{s})}|} \cdot \frac{1}{|V_{G(\mathbf{s}})|}
\end{equation*}
as the function that incentivizes picking labels at least as high as the rest and
\begin{equation*}
    f_\mathbf{s}^\downarrow(e) = \frac{|\{e'\in E_{G(\mathbf{s})}, \lambda_{G(\mathbf{s})}(e') < \lambda_{G(\mathbf{s})}(e)\}|}{|E_{G(\mathbf{s})}|} \cdot \frac{1}{|V_{G(\mathbf{s}})|}
\end{equation*}
as the function that incentivizes picking labels at least as low as the rest. Since the first factor is always less than 1, both functions map to $\left[0, \frac{1}{|V|}\right)$. Lastly, we define our penalty functions. If we do not want to assign varying costs to labels, we use the cost function $f^0_\mathbf{s}(e) = 0$. Next, since it is common in literature to only allow positive labels, we define a penalty for choosing labels smaller than 1. Let
\begin{equation*}
    \text{P}_\mathbf{s}^{> 0}(u) = K\cdot |\{(v, l) \in S_u \mid l < 1\}|
\end{equation*}
with $K \geq |V|$ as a sufficiently large constant. Similarly, we can define another penalty function that penalizes agents for buying 2 edges with the same label or buying an edge towards another one with the same label as one already adjacent. Let
\begin{align*}
    \text{P}_\mathbf{s}^{\neq}(u) &= K\cdot \mathbf{1}\left( \exists (v, l) \in S_u, w \neq v \in V: (w, l) \in S_u \right) \\
    & + K\cdot \mathbf{1}\left(\exists (v,l) \in S_u, w \in V: (w, l) \in S_v \right)
\end{align*}
with $K \geq |V|$ as a sufficiently large constant. This forces agents to facilitate proper labeling.

Let $\mathcal{R}\in\{<,\leq\}$ denote the chosen reachability notion,
$\mathcal{F}\in\{\uparrow,\downarrow,0\}$ the chosen label cost function,
and $\mathcal{P}\subseteq\{\geq 0,\neq\}$ the chosen set of penalty functions.
A \emph{variant} of the model is specified by the tuple
$\mathcal{M}=(\mathcal{R},\mathcal{F},\mathcal{P})$. For a fixed variant \(\mathcal{M}\), the cost of agent \(u\) under strategy profile \(\mathbf{s}\) is
\begin{align*}
c^{\mathcal{M}}(u,\mathbf{s}) &=
|S_u|
+ \sum_{(v,l)\in S_u} f^{\mathcal{F}}_{\mathbf{s}}(\{u,v\})\\
&+ \sum_{P\in\mathcal{P}} P_{\mathbf{s}}(u)
+ K\cdot \bigl|V \setminus R^{\mathcal{R}}_{G(\mathbf{s})}(u)\bigr|.
\end{align*}
For a variant $\mathcal{M}$ we refer to the PoA and PoS values as $\text{PoA}^{\mathcal{R}, \mathcal{F}}_\mathcal{P}$ and $\text{PoS}^{\mathcal{R}, \mathcal{F}}_\mathcal{P}$ respectively.

\section{Uniform Label cost}
\label{chap:flexiblelabel}
We start with the simplest model, where agents have complete freedom to decide what edges and labels they buy, as all time steps are priced equally.

\subsection{Non-strict Paths}

We first examine the non-strict model, where paths must be non-strictly increasing in time. We provide a tight PoS bound that also guarantees the existence of equilibria for this model. Additionally, we show a very good PoA bound.

\begin{proposition}
    \hypertarget{prop:posuniform}
    In the non-strict uniform label cost model, $\text{PoS}_{> 0}^{\leq, 0} = 1$.\label{prop:posuniform}
\end{proposition}
\begin{proof}
    Given an arbitrary center vertex $v$, every other vertex can buy a 1-edge towards $v$, creating a 1-label tree. Since all edges have the same label and the graph is a tree, it is temporally connected. As trees are minimally connected graphs, the social optimum is also $n-1$, which implies $\text{PoS}_{> 0}^{\leq, 0} =1$.
\end{proof}


\begin{lemma}
    In the uniform label cost model, $\text{PoA}_{> 0}^{\leq, 0} < 2$.
    \label{lm:upperPoA1}
\end{lemma}
\begin{proof}
    Let $\mathbf{s}$ be a NE with and $v \in V$. Consider a reachability tree $\mathcal{T}_{G(\mathbf{s})}(v)$. Assume another vertex $v'$ has a strategy $S_{v'}$ buying at least two edges $e_1, e_2 \in E$ that are not in $\mathcal{T}_{G(\mathbf{s})}(v)$. A different strategy $S'_{v'} = (S_{v'} \setminus \{e_1, e_2\}) \cup \{(v', v)\}$ with $\lambda_{v'}((v', v\}) = 1$ has lower costs as $v$ reaches every other vertex using $\mathcal{T}_{G(\mathbf{s})}(v)$ and thus, $v'$ reaches everyone by $(v', v)$ and the paths in $\mathcal{T}_{G(\mathbf{s})}(v)$.
    This limits the number of edges to $n-1$ in $\mathcal{T}_{G(\mathbf{s})}(v)$ plus 1 each per the $n-1$ vertices in $V_\mathbf{s} \setminus \{v\}$. Further, consider a vertex $v' \in \mathcal{T}_{G(\mathbf{s})}(v)$ such that $\lambda((v, v'\})\leq \lambda ((v, v''\})$ for all $v'' \in \mathcal{T}_{G(\mathbf{s})}(v)$. It does not need to buy any edges outside of $\mathcal{T}_{G(\mathbf{s})}(v)$ as it can reach the root at the lowest time and then everyone else from there. Thus,
    \begin{equation*}
        \text{PoA}_{> 0}^{\leq, 0}\leq \frac{2(n-1) - 1}{n-1} < 2.\qedhere
    \end{equation*}
\end{proof}

\begin{theorem}
    \hypertarget{prop:poauniform}
    In the non-strict uniform label cost model, $\text{PoA}_{> 0}^{\leq, 0} =  2 - \mathcal{O}\left(\frac{1}{\sqrt{n}}\right)$
    \label{thm:unilabel2}
\end{theorem}
\begin{proof}
We can construct the following NE. Let $G = (V_\mathbf{s}, E_\mathbf{s}, \lambda_\mathbf{s})$ be a temporal graph with $|V_\mathbf{s}| = k^2$ for any $k \geq 3$. We partition $V_\mathbf{s}$ into sets $V_i = \{v_i^1, v_i^2, \dots v_i^n\}$. Each $v_i^1$ buys a 1-edge towards $v_i^2$ and every $v_i^j$ for $3 \leq j \leq n$ buys a 1-edge towards $v_i^{j-1}$. Each $v_i^2$ buys $k-1$ many 2-edges towards a $v_{i'}^{j'}$ with $i' \neq i$ and $j' \neq 2$ such that no two vertices buy an edge towards the same vertex, which is possible as there are $k$ buying vertices and $k(k-1)$ vertices that can be bought towards (see \Cref{fig:grid}).

We first show that the graph is temporally connected. Each vertex inside each $V_i$ can reach the other vertices of the partition via 1-edges. Further, each vertex besides the 2nd one has a 2-edge bought towards it by a distinct $v_j^2$, thus each vertex in $V_i$ can reach all of them. As each of them also buys $k-1$ edges towards distinct edges that do not buy 2-edges, each $v\in V_i$ can reach $k + k(k-1) = k^2= n$ vertices.

We now show that it is an equilibrium. Each vertex $v$ that is buying a 1-edge needs to buy this edge, since otherwise it has exactly one other incident 2-edge, bought by some $v_i^2$, which in turn has only $k-1$ incident 2-edges and no vertex in buys towards has a different 2-edge. Thus, $v$ would only reach $k < n-1$ vertices, and they cannot improve their strategy. Each vertex that buys 2-edges, buys edges towards vertices that only have incident 1-edges from disjunct partitions. As no two partitions have a 1-edge between each other and no 2-edges exist between two vertices where none is the 2nd of a partition, it is not possible to buy a single edge to reach two such vertices. Thus, they have to buy $k-1$ edges to reach the vertices of the other partitions that do not have a 2-edge bought towards them. As each of the $k$ partitions has $k-1$ edges and each of the $k$ second vertices buy $k-1$ edges
     \begin{align*}
         \text{PoA}_{> 0}^{\leq, 0} &\geq \frac{2k(k-1)}{k^2-1} = \frac{2\sqrt{n}(\sqrt{n}- 1)}{n-1} 
         \\& = \frac{2n - \sqrt{n}}{n-1} = 2 - \mathcal{O}\left(\frac{1}{\sqrt{n}}\right)
     \end{align*}
    By $\Cref{lm:upperPoA1}$ this is tight up to $o(1)$.
\begin{figure}
    \centering
    \begin{tikzpicture}[>=stealth, node distance=1cm, font=\small,  scale=0.8]
    
    \tikzstyle{v}=[circle, draw, inner sep=2pt, minimum size=6mm]

    \node[v] (v11) at (0,0) {$v_1^1$};
    \node[v] (v12) at (0,-1.2) {$v_1^2$};
    \node[v] (v13) at (0,-2.4) {$v_1^3$};

    \node[v] (v21) at (3,0) {$v_2^1$};
    \node[v] (v22) at (3,-1.2) {$v_2^2$};
    \node[v] (v23) at (3,-2.4) {$v_2^3$};

    \node[v] (v31) at (6,0) {$v_3^1$};
    \node[v] (v32) at (6,-1.2) {$v_3^2$};
    \node[v] (v33) at (6,-2.4) {$v_3^3$};

    \draw[->, thick] (v11) -- (v12);
    \draw[->, thick] (v13) -- (v12);
    
    \draw[->, thick] (v21) -- (v22);
    \draw[->, thick] (v23) -- (v22);
    
    \draw[->, thick] (v31) -- (v32);
    \draw[->, thick] (v33) -- (v32);
    
    \draw[->, dashed] (v12) to[] (v21);
    \draw[->, dashed] (v12) to[] (v33);

    \draw[->, dashed] (v22) to[] (v31);
    \draw[->, dashed] (v22) to[] (v13);

    \draw[->, dashed] (v32) to[] (v11);
    \draw[->, dashed] (v32) to[] (v23);

    \node at (0,1) {$V_1$};
    \node at (3,1) {$V_2$};
    \node at (6,1) {$V_3$};
    \end{tikzpicture}

    \caption{NE for $k = 3$. Dashed edges are labeled 2 and solid edges are labeled 1}
    \label{fig:grid}
\end{figure}
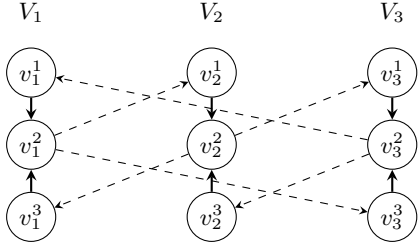
\end{proof}

\subsection{Strict Paths}

\label{sec:strict}
We now examine the same label cost function, but we only consider strictly increasing paths. We provide a PoS and a PoA bound. While the PoS is the same for non-strict and strict models, the PoA for the strict model is significantly higher.

\begin{lemma}
    The social optimum for the strict uniform label cost model is $2n - 4$, even under the constraint of proper edge labeling.
    \label{lem:socoptimum}
\end{lemma}
\begin{proof}
    This is a direct implication of a gossip theory result from \cite{casteigts_sharp_2023}, that the smallest possible temporal spanner of a host graph with $n$ vertices has at least $2n - 4$ edges and the same holds for proper edge labeling.
\end{proof}
\begin{theorem}
    \hypertarget{thm:strictuniformpos}
    For strict uniform label cost $ \text{PoS}_{> 0}^{<, 0} = 1$.\label{thm:unilabpos}
\end{theorem}
\begin{proof}
    We can construct a NE $\mathbf{s}$ for all $n \geq 4$. Let $v_a, v_b, v_c, v_d \in V_\mathbf{s}$. They each buy one edge in counterclockwise direction with $\lambda_\mathbf{s}(\{v_a, v_b\}) = \lambda_\mathbf{s}(\{v_c, v_d\}) = n$ and $\lambda_\mathbf{s}(\{v_b, v_c\}) = \lambda(\{v_d, v_a\}) = n+1$ to form an outer ring. Now let $V_e = V_H \setminus \{v_a, v_b, v_c, v_d\}$ be the set of the other vertices. Let $v_e^k \in V_e$ be t $k$th vertex. If $k \mod 2 \equiv 0$, $v_e^k$ buys a $\frac{k}{2}$-edge towards $v_a$ and $v_c$ buys a $(n+1+\frac{k+2}{2})$-edge towards $v_e^k$. Otherwise $v_e^k$ buys a $\left\lfloor\frac{k}{2}\right\rfloor$-edge towards $v_c$ and $v_a$ buys a $(n+1+\frac{k+1}{2})$-edge towards $v_e^k$ (see \Cref{fig:poastrict}).
    
    $v_a$ and $v_b$ can reach each other via their 2-edge $\{v_a, v_b\}$ and then reach $v_c$ and $v_d$ via the 3-edges $\{v_b, v_c\}$ and $\{v_a, v_d\}$ respectively. The same applies to $v_c$ and $v_d$ by symmetry, so $v_a, v_b, v_c, v_d$ can all reach each other. $v_a$ and $v_c$ buy edges towards the vertices in $V_e$ with label greater than $n$, meaning the vertices in the outer ring can reach all of them via $v_a$ and $v_c$ respectively. All $v_e \in V_e$ buy an edge with label less than $n$ as $k < n$, meaning they reach every other vertex via the outer ring. This proves that $G(\mathbf{s})$ is temporally connected.
    
    Next, we prove it is a NE. $v_a$ and $v_c$ need the edges bought towards the vertices in $V_e$ as they otherwise cannot reach them. They cannot buy fewer edges, since each $v_e^k$ they buy an edge towards only has only one other edge with label $\left\lfloor\frac{k}{2}\right\rfloor$. By our construction $v_a$ and $v_c$ have edges bought towards them with labels $1$ to $\left\lfloor\frac{k}{2}\right\rfloor$. Thus, they cannot buy edges with that label and therefore cannot buy an edge that reaches 2 vertices from $V_e$ that they currently buy an edge towards.
    
    Every $v_e \in V_e$ cannot reach $v_a$ without buying their 1-edge as they otherwise have only incident edges with label $>n+1$ from either $v_a$ or $v_c$ and $v_a$ and $v_c$ have incident edges with labels $\leq n+1$. Lastly, $v_a, v_b, v_c$ and $v_d$ all need the edges they buy in the outer ring to reach each other. Thus, it is a NE. Since the outer ring consists of 4 edges and every $v_e \in V_e$ has 2 incident edges, with $|V_e| = |V_H| - |\{v_a, v_b, v_c, v_d\}| = n-4$, we have
    \begin{equation*}
        \text{PoS}_{> 0}^{<, 0} = \frac{4 + 2(n-4)}{2n-4} = \frac{2n-4}{2n-4} = 1.\qedhere
    \end{equation*}
    \begin{figure}
        \centering
        \begin{tikzpicture}[->, font=\small, scale=0.7]
            \node[circle, draw] (BL) at (0,0) {$v_a$};
            \node[circle, draw] (BR) at (4,0) {$v_b$};
            \node[circle, draw] (TR) at (4,4) {$v_c$};
            \node[circle, draw] (TL) at (0,4) {$v_d$};
        
            \node[circle, draw] (I1) at (1.5,2.5) {$v_{e}^1$};
            \node[circle, draw] (I2) at (2.5,1.5) {$v_{e}^2$};
        
            \draw (BL) -- node[below] {6} (BR);
            \draw (BR) -- node[right] {7} (TR);
            \draw (TR) -- node[above] {6} (TL);
            \draw (TL) -- node[left] {7} (BL);
        
            \draw (I1) -- node[sloped, above] {1} (BL);
            \draw[<-] (I2) -- node[sloped, above] {8} (BL);
        
            \draw (TR) -- node[sloped, above] {8} (I1);
            \draw[<-] (TR) -- node[sloped, above] {1} (I2);
        \end{tikzpicture}
        \caption{PoA in the any-label model for strict paths with $n = 6$}
        \label{fig:poastrict}
    \end{figure}
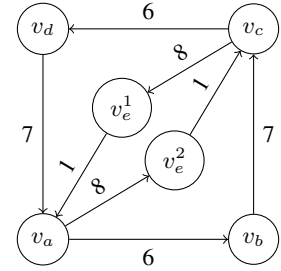
\end{proof}
\begin{theorem}
    \hypertarget{thm:strictuniformpoa}
    In the strict uniform label cost $\text{PoA}_{> 0}^{<, 0} \in \Theta(n)$.\label{thm:anystrictPoA}
\end{theorem}
\begin{proof}
    Consider a strategy graph $G(\mathbf{s})$ that is a clique with $\lambda_{\mathbf{s}}(e) = 1$. Assume a vertex $v$ with a strategy $S_V$ could change to a strategy ${S_v'}$ with lower cost, resulting in a new strategy profile $\mathbf{s}' = \mathbf{s}_{-v} \cup S_v'$ with $|E(\mathbf{s}')| < E(\mathbf{s})$. Then there has to be at least one edge $(v, v') \in G(\mathbf{s})$ that is not in $G(\mathbf{s}')$. Further for all vertices $v''\neq v\colon \lambda_{\textbf{s}'}((v', v''\}) = \lambda_{\textbf{s}}((v', v''\}) = 1$. As there is no direct edge between $v$ and $v'$ in $G(\textbf{s}')$, any path from $v$ to $v'$ would use another edge. For a path from $v$ to $v'$ to exist, there would have to be an edge $e$ incident to $v$ with $\lambda_{\textbf{s}'}(e) < \lambda_{\textbf{s}'}((v', v''\}) = 1$ for any vertex $v''$. However, $\lambda_{\textbf{s}'}$ maps to values at least 1. It follows that there is no strictly increasing path from $v$ to $v'$ and $B'_v$ has higher costs than $S_v$, which is a contradiction to the assumption.
    
    Thus, $\textbf{s}$ is a NE with $\frac{n(n-1)}{2}$ edges. Trivially, there are no NEs with more edges. The social optimum is $2n - 4$, making
    \begin{equation*}
        \text{PoA}_{> 0}^{<, 0} = \frac{n(n-1)}{2(2n-4)} = \frac{n(n-1)}{4(n-2)} \in \Theta(n).\qedhere
    \end{equation*}
\end{proof}

\section{Monotone Label Cost Model}
\label{sec:monotone}
In this section, we consider two different cost functions that give opposite incentives to the agents. In particular, for the cost function $f^\downarrow_\mathbf{s}$, the cost of each label monotonically decreases with its value, while for the cost function $f^\uparrow_\mathbf{s}$, the cost of each label monotonically increases.
\subsection{Non-strict Paths}

\begin{proposition}
     For both $f^\downarrow_\mathbf{s}$ and $f^\uparrow_\mathbf{s}$ as the label cost function, the social optimum is $n-1$.
    \label{prop:labelcostopt}
\end{proposition}
\begin{proof}
    Consider a strategy profile $\mathbf{s}$ where $G(\mathbf{s})$ is a $k$-label tree. Since all labels are equal, the label cost is 0 and since it is a tree, no vertex can buy fewer edges and $\mathbf{s}$ is a NE.
\end{proof}
\begin{theorem}
    \hypertarget{thm:nonstrictmonotonepos}
    For $f^\downarrow_\mathbf{s}$ as the label cost function, $\text{PoS}_{> 0}^{\leq, \downarrow} = \text{PoA}_{> 0}^{\leq, \downarrow} = 1$.
    \label{thm:poaposfdown}
\end{theorem}
\begin{proof}
    Let $\mathbf{s}$ be a NE. Assume there is a vertex $v$ that buys an edge $\{v, u\}$ with label $l$ and there is an edge with label $l'<l$ in $G(\mathbf{s})$. Vertex $v$ can reduce $l$ to $l'$ and still use every edge incident to $u$ with labels $l'' \geq l > l'$. Thus, it can improve its cost, and in a NE, every vertex has to buy an edge of the same label. Since a NE has to be a minimal connected graph $G(\mathbf{s})$, it has to be a $k$-label tree, and by \Cref{prop:labelcostopt}, both the PoA and PoS are 1.
\end{proof}
\begin{proposition}
    \label{nonstrictposup}
    For $f^\uparrow_\mathbf{s}$ as label cost function, $\text{PoS}_{> 0}^{\leq, \uparrow} = 1$.
\end{proposition}
\begin{proof}
    This follows directly from the same argument as the proof for \Cref{thm:poaposfdown}.
\end{proof}

\begin{theorem}
    \label{monotonepoaup}
    For $f_\mathbf{s}^\uparrow$ label cost $\text{PoA}_{> 0}^{\leq, \uparrow} \in \left[2 - \mathcal{O}\left(\frac{1}{\sqrt{n}}\right), 3\right]$.
\end{theorem}
\begin{proof}
Consider the NE construction from \Cref{thm:unilabel2}. The vertices buying 2-edges already do not have additional label cost. Meanwhile, the vertices buying 1-edges could not instead buy a single edge of label 2 or higher, as the size of 2-label trees in the graph is only $k-1$. Thus, since we have shown that no vertex can buy fewer edges and no vertex can improve their label cost, the construction is also a NE in the label cost model with $f_\mathbf{s}^\uparrow$.

Note also that since agents still prioritize reducing the number of edges bought, by \Cref{lm:upperPoA1}, the maximum number of edges is bounded upwards by $2n$. Further, the total label cost is at most 1 per vertex, or at most $n$ in total. Thus,

\begin{equation*}
    \text{PoA}_{> 0}^{\leq, \uparrow} \in \left[2 - \mathcal{O}\left(\frac{1}{\sqrt{n}}\right), 3\right].\qedhere
\end{equation*}
\end{proof}

\subsection{Strict Paths}
\begin{proposition}
    \label{strictmonotonepoa}
    With strict paths it holds for both cost functions $\text{PoA}_{> 0}^{<, \downarrow}\in \Theta(n)$, $\text{PoA}_{> 0}^{<, \uparrow}\in \Theta(n)$
\end{proposition}
\begin{proof}
    The clique of equal labels is a NE since, as established in \Cref{thm:anystrictPoA}, no vertex can buy fewer edges and the label cost is 0 for all vertices. Thus, both PoA are linear.
\end{proof}
\begin{proposition}
    \label{strictdownpos}
    With strict paths, for the $f_\mathbf{s}^\downarrow$ cost function $\text{PoS}_{> 0}^{<, \downarrow}\in \Theta(n)$.
\end{proposition}
\begin{proof}
    If there are at least two different labels, then whichever vertex buys the higher one can replace its label with a lower one and still reach everyone. Thus, in a NE every label is equal. Since paths are strict, each vertex can only reach a single other using an edge, and the only NE is a clique with the same labels.
\end{proof}


\section{Proper Label Model}
\label{sec:proper}
In this section, we consider a model where label costs are always the same, but agents are not allowed to assign a label value $\ell$ to an edge $e$ if edge $e$ already has an adjacent edge $e'$ with label value $\ell$. Because of the following proposition, we can study the strict setting, and the results carry over to the non-strict setting as well.

\begin{proposition}
    \label{propernonstrict}
    PoA and PoS are equal to the strict paths ones if non-strict paths are allowed in the proper label model.
\end{proposition}
\begin{proof}
    Since no two adjacent edges have the same label, no temporal path can contain two subsequent edges of the same label, and for any strategy profile $\mathbf{s}$, $G(\mathbf{s})$ does not contain non-strict paths, meaning allowing non-strict paths does not change the result.
\end{proof}

\subsection{Strict}
\label{sec:nocommon}
\begin{lemma}
    \label{lm:nocommontree}
    Given a temporal graph $G = (V_G, E_G, \lambda_G)$, let $v \in V_G$ and $e_1 \in E_G$ with $v \in e_1$. For $l = \lambda_G(e_1)$ and $l_\text{max} = \max\{ \lambda_G(e') \mid e' \in E_G,\ v \notin e' \}$ the number of vertices reachable from $v$ using strict, and therefore acyclic, temporal paths starting with $e_1$, is at most $\max\left\{2^{l_\text{max}-l}, 1\right\}$.
\end{lemma}

\begin{proof}
    Let $p = (e_1, \dots e_i)$ be a strict acyclic temporal path. We consider the cases $l \geq l_\text{max}$ and $l < l_\text{max}$.\\

    \noindent
\textbf{Case 1:} $l \geq l_\text{max}$. By our assumption $\lambda_G(e_1) \geq l_\text{max} \geq \lambda_G(e_2)$. As $p$ is strict, $e_2$ cannot be in $p$ and $p$ has at most length 1. This makes the number of vertices reached by $p = 1$.\\

    \noindent
    \textbf{Case 2:} $l < l_\text{max}$. Consider the set of labels $\{l+1, \dots, l_\text{max}\}$. Any strict temporal path starting with $e_1$ corresponds to a strictly increasing sequence of labels starting at $l$, with subsequent labels chosen from $[l+1, l_\text{max}]$. Since adjacent edges must have distinct labels, each label sequence corresponds to at most one unique path. The number of such sequences is the number of subsets of $\{l+1, \dots, l_\text{max}\}$, i.e., $2^{l_\text{max}-l}$. Thus, the number of reachable vertices is at most $2^{l_\text{max} - l}$.
\end{proof}


\begin{corollary}
    Any temporal spanner in the proper Label model has lifetime $t \geq \log n$.
\end{corollary}
\begin{proof}
    Assume the graph had lifetime $\leq \log n - 1$. Then the maximum number  of other vertices $N$ a vertex could reach by \Cref{lm:nocommontree} would be
    \begin{align*}
        N \leq \sum_{i = 1}^{\log n - 1} 2^{(\log n - 1) - i} = \sum_{i = 1}^{\log n - 1} 2^{i} = \frac{n}{2} - 1.
    \end{align*}
    This is a contradiction as each vertex needs to reach $n-1$ vertices in total. Thus, it follows that $t \geq \log n$.
\end{proof}

\begin{theorem}
    \label{properpos}
    For proper labels $\text{PoS}_{> 0, \neq}^{<, 0} = 1$.
\end{theorem}
\begin{proof}
    Consider the NE construction from \Cref{thm:unilabpos}. The outer ring has proper labeling and the edges towards and from the center vertices are labeled in ascending order such that they are either strictly smaller or strictly larger than the edges in the outer ring. Thus, the construction is also a NE in the proper label model and $\text{PoS}_{> 0, \neq}^{<, 0} = 1$.
\end{proof}
\begin{theorem}
    \label{properpoa}
    For proper labels $\text{PoA}_{> 0, \neq}^{<, 0} \in \Omega(\log n)$.
\end{theorem}
\begin{proof}
    Consider a $k$-dimensional hypercube $Q_k = (V_k, E_k)$. As shown in \cite{kempe_connectivity_2002}, there exists a labeling $\lambda\colon E_k \mapsto[\log n]$, such that the resulting temporal graph $G = (V_k, E_k, \lambda)$ is a temporal spanner (seen in \Cref{fig:hypercube}). Now we will prove it is also a NE. By \Cref{lm:nocommontree}, with $l_\text{max} = \log n$, each vertex can reach at most $\sum_{i = 1}^{\log n} 2^i = n - 1$ other vertices. As any edge labeled with $l > \log_n = l_\text{max}$ would reach only a single vertex, there is no smaller set of different labels that can also reach $n-1$ vertices
    and thus no agent can improve their strategy. As a hypercube with $n$ vertices has $\Theta(n\log n)$ edges and the social optimum is $2n-4$ following from \Cref{lem:socoptimum},
        \begin{equation*}
            \text{PoA} \geq \frac{\Theta(n\log n)}{\Theta(n)} = \Theta(\log n).\qedhere
        \end{equation*}
    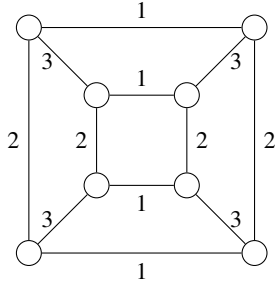
\begin{figure}
    \centering
    \begin{tikzpicture}[scale=0.6, font=\small, baseline=(current bounding box.center)]

    \node[circle, draw] (A) at (0,0) {};
    \node[circle, draw] (B) at (2,0) {};
    \node[circle, draw] (C) at (2,2) {};
    \node[circle, draw] (D) at (0,2) {};
    
    \node[circle, draw] (E) at (-1.5,-1.5) {};
    \node[circle, draw] (F) at (3.5,-1.5) {};
    \node[circle, draw] (G) at (3.5,3.5) {};
    \node[circle, draw] (H) at (-1.5,3.5) {};
    \draw (A) -- node[below] {1} (B);
    \draw (B) -- node[right] {2} (C);
    \draw (C) -- node[above] {1} (D);
    \draw (D) -- node[left] {2} (A);
    
    \draw (E) -- node[below] {1} (F);
    \draw (F) -- node[right] {2} (G);
    \draw (G) -- node[above] {1} (H);
    \draw (H) -- node[left] {2} (E);
    
    \draw (A) -- node[left] {3} (E);
    \draw (B) -- node[right] {3} (F);
    \draw (C) -- node[right] {3} (G);
    \draw (D) -- node[left] {3} (H);
    \end{tikzpicture}
    \caption{Flattened 3D hypercube, NE for $n = 8$.}
    \label{fig:hypercube}
    \end{figure}
\end{proof}

\section{Arbitrary Low Label Model}
\label{sec:nolowest}
In this section, the cost of each label is the same, but we impose a lower bound on the value that an agent can place. While the PoA and PoS bounds are exactly the same as the uniform label cost model in the non-strict model, we get a big improvement on the PoA for the strict model.

\subsection{Non-strict Paths}
\begin{proposition}
\label{nonstrictarbitrary}
    The PoA and PoS results are identical to the uniform label cost model for non-strict paths.
\end{proposition}
\begin{proof}
    Let $\lambda_\text{min}$ be the minimal label of a strategy profile $\mathbf{s}$. Since paths are non-strict, buying an edge with a label smaller than $\lambda_\text{min}$ is equally as good for an agent as simply buying an edge with label $\lambda_\text{min}$. Thus, both PoA and PoS results remain the same as there is no additional improving move.
\end{proof}

\subsection{Strict Paths}

\begin{theorem}
    \label{arbitrarypoa}
    In the Arbitrary Low Label model with strict paths $\text{PoA}_{}^{<, 0} \leq 1 + \frac{1}{n-2}$.
\end{theorem}
\begin{proof}
        The proof works the same way as the proof for \Cref{lm:upperPoA1}, only with the social optimum now being $2n-4$ instead of $n-1$.
\end{proof}
\noindent

\begin{theorem}
\label{arbitrarypos}
    Up to 6 vertices, there exist NEs in the Arbitrary Low Label model with $\text{PoS}_{}^{<, 0} = 1$.
\end{theorem}
\begin{proof}
    We can construct equilibria for up to 6 vertices, displayed in \Cref{fig:nolowest}. For $n \leq 4$ every vertex buys at most 1 edge and it is easy to verify that selling its edge would break its connectivity. Now consider graphs with $n \geq 5$, the vertices that buy 2 edges. They buy an edge towards the top and bottom vertex respectively, as they could only reach it by buying an edge towards it or the adjacent vertex, as they only have an incident 0-edge and no adjacent edges with a smaller label. Further, they need to reach the vertex that they buy the 2-edge towards and cannot reach it without buying an edge, as there is no temporal path towards it.
    
    As the social optimum is at least $2n-4$, no smaller temporal spanners exist for $n > 3$. For $n \leq 2$, trivially, no smaller spanners can exist. Lastly, for $n = 3$, since temporal paths are strict, suppose we have only 2 bought edges and without loss of generality $\lambda_\mathbf{s}(\{u, v\}) \leq \lambda_\mathbf{s}(\{v, w\})$. Then $w$ cannot reach $v$, which means no temporal spanner with 2 edges exists. This makes $\text{PoS}=1$ for $n \leq 6$.
    \begin{figure}
    \centering
    
    \newcommand{\graphcell}[2]{
    \begin{tabular}[c]{@{}c@{}}
    \begin{tikzpicture}[->, scale=0.7, font=\small, baseline=(current bounding box.center)]
    #1
    \end{tikzpicture}\\
    #2
    \end{tabular}
    }
    
    \begin{tabular}[c]{>{\centering\arraybackslash}m{0.3\columnwidth} >{\centering\arraybackslash}m{0.3\columnwidth} >{\centering\arraybackslash}m{0.3\columnwidth}}
    \graphcell{
    \node[circle, draw] (A) at (0,0) {};
    }{$n = 1$}
    &
    \graphcell{
    \node[circle, draw] (A) at (0,0) {};
    \node[circle, draw] (B) at (2,0) {};
    \draw (A) -- node[above] {1} (B);
    }{$n = 2$}
    &
    \graphcell{
    \node[circle, draw] (A) at (0,0) {};
    \node[circle, draw] (B) at (2,0) {};
    \node[circle, draw] (C) at (1,1.73) {};
    \draw (A) -- node[below] {1} (B);
    \draw (B) -- node[right] {1} (C);
    \draw (C) -- node[left] {1} (A);
    }{$n = 3$}
    \end{tabular}
    
    \vspace{1cm}
    
    \begin{tabular}[c]{>{\centering\arraybackslash}m{0.3\columnwidth} >{\centering\arraybackslash}m{0.3\columnwidth} >{\centering\arraybackslash}m{0.3\columnwidth}}
    \graphcell{
    \node[circle, draw] (A) at (0,0) {};
    \node[circle, draw] (B) at (2,0) {};
    \node[circle, draw] (C) at (2,2) {};
    \node[circle, draw] (D) at (0,2) {};
    \draw (A) -- node[below] {1} (B);
    \draw (B) -- node[right] {2} (C);
    \draw (C) -- node[above] {1} (D);
    \draw (D) -- node[left] {2} (A);
    }{$n = 4$}
    &
    \graphcell{
    \node[circle, draw] (A) at (0,0) {};
    \node[circle, draw] (B) at (2,0) {};
    \node[circle, draw] (C) at (2,2) {};
    \node[circle, draw, very thick] (D) at (0,2) {};
    \node[circle, draw] (E) at (1,3) {};
    \draw (A) -- node[below] {1} (B);
    \draw (B) -- node[right] {2} (C);
    \draw (C) -- node[above] {1} (D);
    \draw (D) -- node[left] {2} (A);
    \draw (E) -- node[above right] {0} (C);
    \draw (D) -- node[above left] {3} (E);
    }{$n = 5$}
    &
    \graphcell{
    \node[circle, draw] (A) at (0,0) {};
    \node[circle, draw, very thick] (B) at (2,0) {};
    \node[circle, draw] (C) at (2,2) {};
    \node[circle, draw, very thick] (D) at (0,2) {};
    \node[circle, draw] (E) at (1,3) {};
    \node[circle, draw] (F) at (1,-1) {};
    \draw (A) -- node[below] {1} (B);
    \draw (B) -- node[right] {2} (C);
    \draw (C) -- node[above] {1} (D);
    \draw (D) -- node[left] {2} (A);
    \draw (E) -- node[above right] {0} (C);
    \draw (D) -- node[above left] {3} (E);
    \draw (B) -- node[below right] {3} (F);
    \draw (F) -- node[below left] {0} (A);
    }{$n = 6$}
    \end{tabular}
    
    \caption{Nash equilibria for $n$ between 1 and 6. Vertices that buy 2 edges are highlighted.}
    \label{fig:nolowest}
    \end{figure}
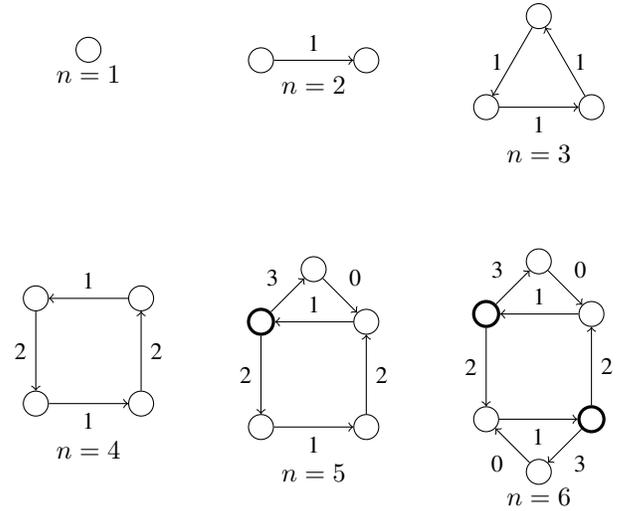
\end{proof}

\section{Outlook and Open Questions}
\label{sec:conclusion}

In this paper, we extended the model introduced by \cite{bilo_temporal_2023} to allow agents to choose the labels assigned to each edge and analyzed different label cost functions under two reachability criteria. There are two immediate open questions from our work: (i) What is the PoS in the Proper Labeling model with non-strict reachability, and (ii) What is the PoS in the Monotone Label model with strict reachability?

At the same time, our model still has some simplifying assumptions that would be worth lifting. For example, we still assume the host network is complete, which may not be the case in real-world scenarios. Additionally, we currently assume that travel time is one for every edge, and it would be interesting to consider variable travel time costs. Finally, we can consider scenarios in which the agents may have more intricate objectives than just reaching and being reached by everyone else; they may want to minimize the hop distance of their trip and/or the time it takes to reach their destination.





\bibliographystyle{named}
\bibliography{sample}

\end{document}